%% file: main.tex
\documentclass[11pt]{article}

\usepackage[english]{babel}
\usepackage[utf8]{inputenc}
\usepackage[T1]{fontenc}
\usepackage{utopia}
\usepackage[margin=1in]{geometry}
\usepackage{amsmath}
\usepackage{amsthm,amssymb}
\usepackage{graphicx}
\usepackage{wrapfig}
\usepackage[colorinlistoftodos]{todonotes}
\usepackage[colorlinks=true, allcolors=blue,pagebackref]{hyperref}
\usepackage{cleveref}
\usepackage{xcolor}
\usepackage{cite}
\usepackage{algorithm}
\usepackage[noend]{algpseudocode}

\newcommand{\E}{\mathbb{E}}
\newcommand{\Z}{\mathbb{Z}}
\newcommand{\R}{\mathbb{R}}
\newcommand{\AHTP}{\textsf{Algorithmic Hypergraph Tur\'{a}n Problem}}

\newtheorem{theorem}{Theorem}
\newtheorem{lemma}[theorem]{Lemma}
\newtheorem{definition}[theorem]{Definition}
\newtheorem{problem}[theorem]{Problem}

\newtheorem*{T1}{Theorem~\ref{thm:tent}}
\title{\textbf{Approximate Hypergraph Vertex Cover \\ and generalized Tuza's conjecture}}
\author{Venkatesan Guruswami\thanks{{\tt venkatg@cs.cmu.edu}. Research supported in part by NSF grant CCF-1908125 and a Simons Investigator award.}  \and Sai Sandeep\thanks{{\tt spallerl@cs.cmu.edu}. Research supported in part by NSF grants CCF-1563742 and CCF-1908125.}}
\date{Computer Science Department \\ Carnegie Mellon University \\ Pittsburgh, PA 15213}
\renewcommand{\epsilon}{\varepsilon}
\begin{document}
\maketitle
\thispagestyle{empty}
\begin{abstract}

A famous conjecture of Tuza states that the minimum number of edges needed to cover all the triangles in a graph is at most twice the maximum number of edge-disjoint triangles.  This conjecture was couched in a broader setting by Aharoni and Zerbib who proposed a hypergraph version of this conjecture, and also studied its implied fractional versions. 
We establish the fractional version of the Aharoni-Zerbib conjecture up to lower order terms. Specifically, we give a factor $t/2+ O(\sqrt{t \log t})$ approximation based on LP rounding for an algorithmic version of the hypergraph Tur\'{a}n problem (AHTP). The objective in AHTP is to pick the smallest collection of $(t-1)$-sized subsets of vertices of an input $t$-uniform hypergraph such that every hyperedge contains one of these subsets.

\smallskip
Aharoni and Zerbib also posed whether Tuza's conjecture and its hypergraph versions could follow from non-trivial duality gaps between vertex covers and matchings on hypergraphs that exclude certain sub-hypergraphs, 
for instance, a ``tent" structure that cannot occur in the incidence of triangles and edges. We give a strong negative answer to this question, by exhibiting tent-free hypergraphs, and indeed $\mathcal{F}$-free hypergraphs for any finite family $\mathcal{F}$ of excluded sub-hypergraphs, whose vertex covers must include almost all the vertices.

\iffalse
hold more generally on hypergraphs that exclude certain sub-hypergraphs, for instance a ``tent" structure that cannot occur in the incidence of triangles and edges. We give a strong negative answer to this question, refuting the hypergraph version of Tuza's conjecture on tent-free hypergraphs, and indeed for $\mathcal{F}$-free hypergraphs for any finite family $\mathcal{F}$ of excluded sub-hypergraphs. 
\fi

\smallskip
The algorithmic questions arising in the above study can be phrased as instances of vertex cover on \emph{simple} hypergraphs, whose hyperedges can pairwise share at most one vertex. 
We prove that the trivial factor $t$ approximation for vertex cover is hard to improve for simple $t$-uniform hypergraphs. However, for set cover on simple $n$-vertex hypergraphs, the greedy algorithm achieves a factor $(\ln n)/2$, better than the optimal $\ln n$ factor for general hypergraphs.

\end{abstract}
\newpage
\input{introduction}
\input{prelims}
\input{algorithm}
\input{structure}

\input{simple}

\section*{Acknowledgment}
We thank Vincent Cohen-Addad for telling us about the Johnson Coverage Hypothesis and its connection to hardness of clustering and for sharing a copy of \cite{CKL20}.

\bibliography{ref}
\bibliographystyle{alpha}
\end{document}

%% file: introduction.tex
\section{Introduction}
\label{sec:intro}
The relationship between minimum vertex covers and maximum matchings of graphs and hypergraphs is a fundamental and well-studied topic in combinatorics and optimization. Even though the worst-case factor $t$ gap between the two parameters cannot be improved on arbitrary $t$-uniform hypergraphs, there are some interesting special cases where the ratio between these quantities is smaller. A classic example of this phenomenon is the K\"{o}nig's theorem on bipartite graphs, where the sizes of minimum vertex covers and maximum matchings are equal. 

For the case of $t=3$, a notorious open problem capturing this gap on special $3$-uniform hypergraphs is Tuza's conjecture\cite{Tuza81, Tuza90}, which states that in any graph, the number of edges required to hit all triangles is at most twice the maximum number of edge-disjoint triangles. 
For a hypergraph $H$, let us denote by $\tau(H)$ and $\nu(H)$ the sizes of the minimum vertex cover and maximum matching respectively.
Tuza's conjecture is then equivalent to the statement 
$\tau(H)\leq 2 \cdot \nu (H)$ for any $3$-uniform hypergraph 
$H$ obtained by taking the edges of a graph $G$ as its vertices, and the triangles of $G$ as its (hyper)-edges. (Taking $G=K_4$ shows that the factor $2$ is best possible.) The conjecture has been verified for various classes of graphs such as graphs without $K_{3,3}$-subdivision~\cite{Kri95},  graphs with maximum average degree less than $7$~\cite{Puleo15}, graphs with quadratic number of edge disjoint triangles~\cite{HaxellR01,Yuster12}, graphs with treewidth at most $6$\cite{BotlerFG19}, and random graphs in the $G_{n,p}$ model~\cite{BCD20,KP20}. On general graphs, the current best upper bound on the ratio is a factor of $2.87$ due to Haxell~\cite{Haxell99}.

Aharoni and Zerbib~\cite{AZ20} introduced an extension
of Tuza's conjecture to hypergraphs of larger uniformity. 
This generalized Tuza's conjecture states that for any $t$-uniform hypergraph $H$, the minimum vertex cover $\tau(H')$ of $H'=H^{(t-1)}$ is at most $\left\lceil\frac{t+1}{2}\right\rceil$ times that of the maximum matching $\nu(H')$. Here, for a $t$-uniform hypergraph $H=(V,E)$, the $(t-1)$-\emph{blown-up hypergraph} $H'=H^{(t-1)}$ is a $t$-uniform hypergraph whose vertices are the set of all $(t-1)$ sized subsets that are contained in at least one edge of $H$, and corresponding to every edge $e$ in $H$, all the $(t-1)$-sized subsets of $e$ form an edge in $H'$.
Tuza's conjecture is a special case of their conjecture when $t=3$ and $H$ has hyperedges corresponding to the triangles in a graph. As is the case with the original Tuza's conjecture, the conjectured value of $\left\lceil\frac{t+1}{2}\right\rceil$ is the best possible gap: when $H$ is the complete $t$-uniform hypergraph on $(t+1)$ vertices, the $(t-1)$-blown-up hypergraph $H'=H^{(t-1)}$ has $\nu(H')=1$ and $\tau(H')=\left\lceil\frac{t+1}{2}\right\rceil$.

\subsection{Fractional Tuza's conjecture and the algorithmic hypergraph Tur\'{a}n problem}
A first step towards non-trivially bounding $\tau(H)$ in terms of $\nu(H)$ for hypergraphs $H$ from some 
structured family of hypergraphs is proving its \emph{fractional version}, i.e., obtaining the same upper bound on the ratio between $\tau(H)$ and $\nu^*(H)$, the fractional maximum matching size. By LP duality, this is equivalent to bounding the ratio between $\tau(H)$ and $\tau^*(H)$, the fractional vertex cover value.
As $\nu(H)\leq \nu^*(H) = \tau^*(H)\leq \tau(H')$ for any hypergraph $H$, establishing the fractional version is a necessary step toward bounding $\tau(H)/\nu(H)$.
Note that understanding the extremal ratio between $\tau$ and $\tau^*$ on a given family of hypergraphs is equivalent to bounding the integrality gap of the natural linear programming relaxation of vertex cover on that class of hypergraphs.

Krivelevich\cite{Kri95} proved the fractional version of Tuza's conjecture that $\tau(H^{(2)})\leq 2\tau^*(H^{(2)})$ for any $3$-uniform hypergraph $H$. A multi-transversal version of Krivelevich's result is proved in a recent work~\cite{ChalermsookKSU20}.
In this work, we prove the fractional version of the generalized Tuza's conjecture (upto $o(t)$ factors), establishing a non-trivial upper bound on the LP integrality gap for $(t-1)$-blown-up hypergraphs.
\begin{theorem}
\label{thm:fractional-tuza}
For any $t$-uniform hypergraph $H$, $\tau\left(H'\right) \leq \left(\frac{t}{2}+2\sqrt{t\ln t} \right) \tau^* \left(H'\right)$, where $\tau(H')$ and $\tau^*(H')$ are respectively the size of the minimum vertex cover and minimum fractional vertex cover of the blown-up hypergraph $H'=H^{(t-1)}$.  Furthermore, there is an efficient algorithm to approximate vertex cover on $(t-1)$-blown-up hypergraphs within a $\frac{t}{2}+2\sqrt{t\ln t}$ factor.
\end{theorem} 

The vertex cover problem on $(t-1)$-blown-up hypergraphs is also intimately connected to the famous Hypergraph Tur\'{a}n Problem~\cite{Turan41,Turan61} in extremal combinatorics.
In the Hypergraph Tur\'{a}n Problem, the goal is to find the minimum size of a family $\mathcal{F}\subseteq \binom{[n]}{t-1}$ of subsets of $[n]$ with cardinality $(t-1)$ such that for every subset $S$ of $[n]$ of size $t$, there exists a set $T \in \mathcal{F}$ such that $T$ is a subset of $S$. The best known upper bound is due to~\cite{Sidorenko97}: there exists a family $\mathcal{F} \subseteq \binom{[n]}{t-1}$ of size $O(\frac{\log t}{t})\binom{n}{t-1}$ such that for every subset $S$ of $[n]$ of size $t$, there exists a subset $T\in \mathcal{F}$ such that $T$ is contained in $S$. 
On the other hand, the lower bound situation is rather dire, with only second-order improvements~\cite{chung99, LuZ09} over the trivial $\frac{1}{t}\binom{n}{t-1}$ lower bound.
%Towards this, de Cain~\cite{decain} conjectured that $t\binom{n}{t-1}^{-1}|\mathcal{F}| \rightarrow \infty$ for any such $\mathcal{F}\subseteq\binom{[n]}{t-1}$.
%such that for every $S \in \binom{[n]}{t}$, there exists $T\in \mathcal{F}$ such that $T \subseteq S$. 

Note that the Hypergraph Tur\'{a}n Problem is precisely the minimum vertex cover problem on $H^{(t-1)}$ when $H$ is the \emph{complete} $t$-uniform hypergraph.  
Thus, for a general hypergraph $H$, finding vertex covers on the blown-up hypergraph $H^{(t-1)}$ can be viewed as an algorithmic version of the Hypergraph Tur\'{a}n problem.
\begin{problem}(\AHTP \ (\textsf{AHTP}))
Given a $t$-uniform hypergraph $H=(V=[n],E)$, find the minimum size of a family $\mathcal{F}\subseteq \binom{[n]}{t-1}$ of subsets of $V$ of size $(t-1)$ such that for every hyperedge $e \in E, $ there exists $T\in \mathcal{F}$ such that $T$ is a subset of $e$.
\end{problem}
The problem is a generalization of the minimum vertex cover on graphs, which corresponds to the case $t=2$.  
As $\textsf{AHTP}$ can be cast as a vertex cover problem on $t$-uniform hypergraphs, there is a trivial factor $t$ approximation algorithm for this problem. 
We prove Theorem~\ref{thm:fractional-tuza} by obtaining an improved algorithm for $\textsf{AHTP}$ based on rounding the standard LP relaxation on $H'=H^{(t-1)}$. 

We now briefly describe this rounding approach. 
First, using threshold rounding, we argue that one may focus on the case when 
the LP solution does not have any variables that are assigned values greater than $\frac{2}{t}$. 
Let $S$ be the set of vertices of $H'$ that are assigned non-zero LP value. 
The thresholding procedure ensures that every hyperedge $e \in E(H')$ intersects with $S$ in at least $\frac{t}{2}$ vertices. 
We can bound the cardinality of $S$ from above by $t \cdot \textsf{OPT}$ using the dual matching LP, where $\textsf{OPT}$ is the cost of the optimal LP solution. Our goal then becomes finding a vertex cover of size at most about $\frac{|S|}{2}$. We achieve this by a color-coding technique: we randomly assign a color from $\{0,1\}$ to each vertex of $H$ independently. Most edges of $H$ are almost balanced under this coloring, in the sense that each color appears at least $t/2-o(t)$ times. We then use this balance property to find a small vertex cover in $H'$.

%Once we assign the colors to the vertices of $H$, we argue that most of the hyperedges of $H$ satisfy a certain uniformity property in the sense that each color appears at least $0.99 \left(\frac{t}{2}\right)$ times in the edge. We then use this uniformity property to find a small vertex cover in $H'$. 

\subsection{Vertex cover vs. matching and excluded sub-hypergraphs}

The generalized Tuza's conjecture concerns the relationship between $\tau$ and $\nu$ on the $(t-1)$-blown-up hypergraphs.
There have been some works in the literature on the gap between $\tau$ and $\nu$ on other structured class of hypergraphs. 
An outstanding result of this type is Aharoni's proof~\cite{Aharoni01} that $\tau(H)\leq 2\nu(H)$ for all tripartite $3$-uniform hypergraphs $H$. 
Aharoni and Zerbib~\cite{AZ20} asked if there is a structural explanation that unites the generalized Tuza's conjecture and the above result---for example, does the exclusion of a certain substructure in the hypergraph $H$ imply better gaps between $\tau(H)$ and $\nu(H)$. A particular substructure they studied is the ``tent'' subhypergraph (\Cref{fig:tent}). %
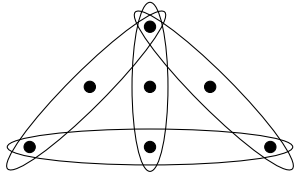
\begin{figure}
    \centering
\begin{tikzpicture}[scale=0.4]
\draw (0,0) node[fill,circle, draw, fill=black,inner sep=0,minimum size=0.15cm] {};
\draw (0,-2) node[fill,circle, draw, fill=black,inner sep=0,minimum size=0.15cm] {};
\draw (0,-4) node[fill,circle, draw, fill=black,inner sep=0,minimum size=0.15cm] {};
\draw (-4,-4) node[fill,circle, draw, fill=black,inner sep=0,minimum size=0.15cm] {};
\draw (4,-4) node[fill,circle, draw, fill=black,inner sep=0,minimum size=0.15cm] {};
\draw (-2,-2) node[fill,circle, draw, fill=black,inner sep=0,minimum size=0.15cm] {};
\draw (2,-2) node[fill,circle, draw, fill=black,inner sep=0,minimum size=0.15cm] {};

\draw[rotate=45] (-3,0)ellipse (105pt and 17pt);
\draw[rotate=135] (-3,0)ellipse (105pt and 17pt);
\draw (0,-4)ellipse (135pt and 17pt);
\draw (0,-2)ellipse (17pt and 80pt);
\end{tikzpicture}
    \caption{The $3$-tent}
    \label{fig:tent}
\end{figure}

They observed that both tripartite hypergraphs and $2$-blown-up $3$-uniform hypergraphs cannot contain the tent as a subhypergraph, and asked whether a generalization of Tuza's conjecture might hold for
$3$-uniform hypergraphs that exclude tents. If this is the case, it could give a common structural explanation of the existence of small vertex covers in $3$-uniform hypergraphs.

In this paper, we answer this question in the negative. We prove that that there are hypergraphs $H$ on $n$ vertices that exclude tents with $\tau(H) \ge (1-o(1))n$. Since $\nu(H) \le n/3$ trivially, this shows that the ratio $\tau/\nu$ can approach $3$ on tent-free $3$-uniform hypergraphs, and the extension of Tuza's conjecture as raised in \cite{AZ20} does not hold.
More generally, one might ask if there is some collection of hypergraphs which are excluded from blown-up hypergraphs whose absence implies a non-trivial gap between $\tau$ and $\nu$. 
In fact, we prove a stronger statement showing that there 
is no $3$-uniform hypergraph family $\mathcal{F}$ (that is absent from blown-up hypergraphs) whose exclusion alone could imply Tuza's conjecture. Our result applies for larger uniformity $t$ and the fractional version of Tuza's conjecture. 
\begin{theorem}
\label{thm:tent}
For every $\epsilon >0$ and every finite family of $t$-uniform hypergraphs $\mathcal{F}$ such that no hypergraph from $\mathcal{F}$ appears in any $(t-1)$-blown-up hypergraph $H'=H^{(t-1)}$, there is a $t$-uniform hypergraph $T$ such that $T$ does not contain any hypergraph from $\mathcal{F}$ but $\tau(T)\geq (t-\epsilon)\nu(T)$ (and in fact $\tau(T) \geq (t-\epsilon) \tau^*(T)$).
\end{theorem}

The above result rules out the possibility of a ``local'' proof of Tuza's conjecture.
Our construction is a probabilistic one, first sampling each edge of the hypergraph independently with certain probability, and then removing all the copies of hypergraphs in $\mathcal{F}$. Using the fact that the family $\mathcal{F}$ satisfies certain sparsity requirements~\cite{FM08,BFM10}, we can conclude that there is no large independent set in this construction.

We also provide an \emph{explicit construction} that answers the tent-free question of \cite{AZ20}: our counterexample is the hypergraph $T$ with vertex set $[3]^n$ for large enough $n$ and edges being the set of combinatorial lines. By the density Hales Jewett Theorem~\cite{FK91,poly12}, there is no large independent set in $T$, and using the structure of combinatorial lines, we can prove that $T$ does not have any tent.

\subsection{Vertex cover and set cover on simple hypergraphs}
As mentioned earlier, \textsf{AHTP} is a special case of vertex cover on $t$-uniform hypergraphs. In fact, the blown-up hypergraph $H^{(t-1)}$ is a \textit{simple}\footnote{Simple hypergraphs are also referred to as linear hypergraphs.} hypergraph: any two edges intersect in at most one vertex.
This is simply because any two distinct $t$-sized subsets of $[n]$ intersect in at most one $(t-1)$-sized subset. 
%A hypergraph in which any two edges intersect in at most one vertex is known as a \textit{simple} hypergraph.\footnote{Simple hypergraphs are also referred to as linear hypergraphs.}
Simple hypergraphs have been well studied in Graph Theory, especially in the context of Erd\H{o}s-Faber-Lov\'{a}sz conjecture~\cite{Erds1981, Erdos1988}, Ryser's conjecture~\cite{FHMW17} and chromatic number of bounded degree hypergraphs~\cite{DukeLR95,FriezeM13}.

A natural question is whether we can obtain an approximation ratio smaller than $t$ for vertex cover on simple hypergraphs. 
However,~\Cref{thm:tent} shows that the natural LP has an integrality gap approaching $t$ on simple, and indeed a lot more structured, hypergraphs. But perhaps there are other algorithms that beat the trivial factor $t$ approximation for this problem. 
We prove that this is not the case, and in fact, vertex cover on simple hypergraphs is as hard as vertex cover on general $t$-uniform hypergraphs. 
\begin{theorem}
\label{thm:simple}
For every $\epsilon >0$, unless $\textsf{NP}\subseteq\textsf{BPP}$, no polynomial time algorithm can approximate vertex cover on simple $t$-uniform hypergraphs within a factor of $t-1-\epsilon$. Under the Unique Games conjecture, the inapproximability factor improves to $t-\epsilon$.
\end{theorem}

We also study the set cover problem on simple set families where any two sets in the family intersect in at most one element. Equivalently, we want to pick the minimum number of edges to cover all vertices in a simple hypergraph. Interestingly, simplicity of the set family helps in getting an improved approximation factor for the set cover---in fact, the greedy algorithm itself delivers such an approximation.

\begin{theorem}
\label{thm:simple-set-cover-intro}
For set cover on simple set systems over a universe of size $n$, the greedy algorithm achieves an approximation ratio $\frac{\ln n}{2} + 1$. Further, there are simple set systems where the greedy is off by a factor exceeding $\frac{\ln n}{2}-1$.
\end{theorem}

The same techniques used in proving~\Cref{thm:simple} also gives an inapproximability factor of $(\ln n)^{\Omega(1)}$ for set cover on simple set families. Interestingly, the dual Maximum Coverage problem, where the goal is to cover as many elements as possible with a specified number of sets, does not become easier on simple set systems and is hard to approximate within a factor exceeding $(1-1/e)$~\cite{CKL20a}, the factor achieved by the greedy algorithm on general set systems.
In \cite{CKL20}, the authors conjecture the hardness of achieving an approximation factor beating $(1-1/e)$ even for the Maximum Coverage version of AHTP, and call this
the 
\emph{Johnson Coverage Hypothesis}. They show that this hypothesis implies strong inapproximability results for fundamental clustering problems like $k$-means and $k$-median on Euclidean metrics. For example, they showed that the hypothesis implies that $k$-median is hard to approximate within a factor of $1.73$ on $\ell_1$ metrics, matching the best hardness factor on general metrics due to Guha and Khuller~\cite{GuhaK98}.

%\vg{Fix citations above; fix gap factor for greedy, and add something about it in the body? Be more precise about clustering implications?}

%\vg{Interestingly, Max Coverage is $(1-1/e+\epsilon)$ hard on simple set systems. Could this be a place to talk about JCH and the clustering connectionj?}

%The authors of~\cite{AZ20} also informally ask if there is a finite list of forbidden substructures that characterizes whether a given $3$-uniform hypergraph $H$ is equal to $G^{(2)}$ for some hypergraph $G$. 
%We make this conjecture formally for general $t$ and show that if an algorithmic version of this conjecture holds, then we can get an improved algorithm for the \textsf{AHTP}. 
%We continue further by proving that a complete characterization of the blown-up hypergraphs would imply improved approximation algorithms to the \textsf{AHTP}. 

%\medskip\noindent\textbf{Other improved hypergraph vertex cover algorithms.}
\subsection{Other improved hypergraph vertex cover algorithms}
Algorithms beating the trivial factor $t$  approximation have been obtained for the vertex cover problem on some other families of $t$-uniform hypergraphs. In his doctoral thesis, Lov\'{a}sz~\cite{Lov75} gave a LP rounding algorithm to obtain a factor $\frac{t}{2}$ approximation for vertex cover on $t$-uniform $t$-partite hypergraphs. This algorithm is shown to be optimal under the Unique Games Conjecture by Guruswami, Sachdeva, and Saket~\cite{GSS15}, and an almost matching NP-hardness is also shown.
Aharoni, Holzman, and Krivelevich~\cite{AHK96} generalized the above algorithmic result to other class of hypergraphs which have a partition of vertices obeying certain size restrictions. 
A factor $\frac{t}{2}$ approximation algorithm has also been obtained on subdense regular $t$-uniform hypergraphs~\cite{cardinal2012}.

For the problem of covering all paths of length $t$ ($t$-Path Transversal), Lee~\cite{Lee19} gave a factor $O(\log t)$ approximation. For covering all copies of the star on $t$ vertices, i.e., $K_{1,t-1}$, a factor $O(\log t)$ approximation is given in \cite{GL-sidma}, and this is tight by a simple reduction from dominating set on degree $t$ graphs. Covering $2$-connected $t$-vertex pattern graphs (in particular $t$-cliques or $t$-cycles) is as hard as general $t$-uniform hypergraph vertex cover~\cite{GL-sidma}.

\subsection{Open problems}
\label{subsec:open-probs}
A number of intriguing questions and directions come to light following our work, and we mention a few of them below.

The most obvious question is whether our algorithm for AHTP can be improved and yield approximation ratios smaller than $t/2$. Can stronger LP relaxations like Sherali-Adams help in this regard? On the hardness side, essentially nothing is known. There is a straightforward approximation preserving reduction from  vertex cover on graphs to AHTP, but this only shows the hardness of beating a factor of $2$. Can one show a better inapproximability factor?  We do not know any good lower bound on the integrality gap of the LP either---for example, we do not know the existence of a hypergraph $H$ for which $\tau(H^{(t-1)})/\tau^*(H^{(t-1)})$ grows with $t$. A natural candidate is the complete $t$-uniform hypergraph on $n$ vertices for which de Cain conjectured~\cite{decain} that in fact, $\tau(H^{(t-1)})/\tau^*(H^{(t-1)})$ grows with $t$. 
However, this is precisely the lower bound of hypergraph Tur\'{a}n problem, and is perhaps very hard to resolve. 
On the algorithmic side, obtaining $o(\log t)$ approximation algorithm for AHTP would lead to improvements on the hypergraph Tur\'{a}n problem: On the complete hypergraph instance, either the algorithm outputs a family $\mathcal{F}\subseteq \binom{[n]}{t-1}$ of size $o\left(\frac{\log t}{t-1}\right)\binom{n}{t}$ that covers every subset of size $t$, or gives a certificate that any such family should have size at least $\omega\left(\frac{1}{t}\right)\binom{n}{t-1}$. In the first case, we get an improvement on the upper bound of hypergraph Tur\'{a}n problem, and in the second case, we resolve de Cain's conjecture. 

%It is interesting to know if there are other hypergraphs for which this ratio grows with $t$ or if the complete hypergraph is  are the best possible constructions. 
%Recall that we do know a lower bound of $\lceil \frac{t+1}{2} \rceil$ on $\tau(H^{(t-1)})/\nu(H^{(t-1)})$, and this is conjectured to be tight.

Similar to the $(t-1)$-blown-up hypergraphs, one can define the $k$-blown-up hypergraph of a $t$-uniform hypergraph---which will be a ${t \choose k}$-uniform hypergraph---and study the vertex cover problem on it. A special case of this problem when $k=2$ is the analog of Tuza's problem for larger cliques, i.e., covering all copies of $t$-cliques in a graph by the fewest possible edges. Our algorithm for AHTP extends to this setting, and in particular gives an algorithm with ratio $t^2/4$ for the $k=2$ case, beating the trivial ${t \choose 2}$ factor (see Section~\ref{subset-(t,2)} for details). Can one achieve a $o(t^2)$ factor algorithm? A simple reduction from vertex cover on $t$-uniform hypergraphs shows an inapproximability factor of $t-O(1)$, but can one show hardness or integrality gaps of $\omega(t)$?

In general, our work brings to the fore challenges about covering graph structures by \emph{edges}, on both the algorithmic and hardness fronts. On the hardness side, we seem to have essentially no techniques to show strong inapproximability results, as the known PCP techniques where one naturally associates vertices with proof locations do not seem to apply.  As mentioned earlier, covering all copies of a $t$-vertex pattern graph $H$ with \emph{vertices} is as hard to approximate as general $t$-uniform hypergraph vertex cover when $H$ is $2$-connected~\cite{GL-sidma}.

%\vg{Is this okay? Can you tighten or spruce it up somehow? Any simple example challenge problem other than the above?}

Our structural results show that the LP integrality gap (and therefore also the vertex cover to matching ratio) remains close to $t$ on hypergraphs that exclude subgraphs absent in $(t-1)$-blown-up hypergraphs, and thus no ``local" proof of Tuza-type conjectures is possible. Are there interesting families of $t$-uniform hypergraphs $\mathcal{F}$ such that vertex cover admits non-trivial approximation (with ratio less than $t$) on $\mathcal{F}$-free $t$-uniform hypergraphs?

What is the optimal approximation factor one can achieve for set cover on simple set systems? Is a $o(\ln n)$ approximation possible? For simple set systems with set sizes bounded by $t$, is there an algorithm with approximation ratio $c \ln t$ for some $c < 1$? We note that the greedy algorithm itself cannot provide such a guarantee, as our tight example for greedy in Theorem~\ref{thm:simple-set-cover-intro} uses sets of size at most $\sqrt{n}$.

For the maximization version of AHTP, where we seek to pick a specified number $(t-1)$-sized subsets to cover the largest number of edges in a $t$-uniform hypergraph, is there an algorithm that beats the $(1-1/e)$ factor (achieved by greedy for the general Max Coverage problem)? The Johnson Coverage Hypothesis of \cite{CKL20} asserts that for any $\epsilon > 0$, a $(1-1/e+\epsilon)$-approximation is hard to obtain for $t$ large enough compared to $\epsilon$.

We have considered covering problems in this work, and there are interesting questions concerning the dual packing problems as well. For instance, what is the approximability of packing \emph{edge-disjoint} copies, of say $t$-cliques, in a graph? This is a special case of the matching problem on $2$-blown-up hypergraphs. 
For the maximum matching problem on general $k$-uniform hypergraphs, also known as $k$-set packing, Cygan~\cite{Cygan13} gave a local search algorithm that achieves an approximation factor of $\frac{k+1}{3}+\epsilon$ for any constant $\epsilon >0$. Can we get better algorithms for the maximum matching problem on blown-up hypergraphs? 
%What about matching on simple hypergraphs? 

On the hardness front, $k$-set packing is inapproximable to a $\Omega(k/\log k)$ factor~\cite{HSS06}. Known inapproximability results for the independent set problem on graphs with maximum degree $k$~\cite{AKS11,Chan16} imply that the maximum matching problem on $k$-uniform \emph{simple} hypergraphs is hard to approximate within a $\Omega\left(\frac{k}{\log ^2 k}\right)$ factor. Could maximum matching on simple hypergraphs be easier to approximate than general hypergraphs?

\medskip\noindent\textbf{Organization of the paper.} In~\Cref{sec:prelims}, we introduce some notation and definitions. In~\Cref{sec:alg}, we describe and analyze our algorithm for AHTP and prove~\Cref{thm:fractional-tuza}.  Then, in~\Cref{sec:structure}, we prove that the analog of (generalized) Tuza's conjecture does not hold based only on local forbidden sub-hypergraph characterizations, proving Theorem~\ref{thm:tent}, and also giving an explicit construction for the tent-free case posed in \cite{AZ20}.
 Finally, in~\Cref{sec:simple}, we consider simple hypergraphs and prove Theorems~\ref{thm:simple} and \ref{thm:simple-set-cover-intro}.
 

%% file: prelims.tex
\section{Preliminaries}
\label{sec:prelims}
\paragraph{Notation.} We use $[n]$ to denote the set $\{1,2,\ldots,n\}$. 
We use $\Z_n$ to denote the set $\{0,1,\ldots, n-1\}$.
For a set $S$ and an integer $1 \leq k \leq |S|$, we use $\binom{S}{k}$ to denote the family of all the $k$-sized subsets of $S$. 
A hypergraph $H'=(V',E')$ is called a subhypergraph of $H=(V,E)$ if $V'\subseteq V$ and $E'\subseteq E'$. 
For a hypergraph $H=(V,E)$, we use $\tau(H), \nu(H)$ to denote the size of the minimum vertex cover and the maximum matching respectively. 
Similarly, we use $\tau^*(H)$ to denote the minimum fractional vertex cover of $H$: 
\[
\tau^*(H) = \min \left\{ \sum_{v\in V}x_v : x_v \in \R_{\geq 0} \,\,\forall v \in V, \sum_{v \in e} x_v \geq 1 \,\,\forall e \in E \right\}
\]

We define the $k$-blown up hypergraph formally:  
\begin{definition}
	For a $t$-uniform hypergraph $G=(V, E)$ and for an integer $ 1 \leq k < t$, we define the $k$-blown up hypergraph $H=G^{(k)}=(V',E')$ as follows: 
	\begin{enumerate}
		\item The vertex set $V' \subseteq \binom{V}{k}$ is the set of all $k$-sized subsets of $V$ that are contained in an edge of $G$: 
		\[
			V' = \left\{ U : U \subseteq V, |U|=k, \exists e \in E : U \subseteq e\right\}
		\]
		\item For every edge $e \in E$, we include in $E'$ all the $k$-sized subsets of $e$, so that
		\[
			E' = \left\{ e' : e'= \binom{e}{k}, e \in E\right\}
		\]
	\end{enumerate}
\end{definition}
We will need the following Chernoff bound: 
% \begin{theorem}(Markov's inequality)
% 	\label{thm:markov}
% 	Let $X$ be a non-negative real random variable with $\E[X]=\mu$. Then for every $a>0$, we have 
% 	\[
% 		\text{Pr}(X\geq a) \leq \frac{\E[X]}{a}
% 	\]
% \end{theorem}
\begin{lemma}(Multiplicative Chernoff bound)
	\label{thm:chernoff}
	Suppose $X_1, X_2, \ldots, X_n$ are independent random variables taking values in $\{0,1\}$. Let $X= X_1 + X_2 + \ldots + X_n$, and let $\mu = \E[X]$. Then, for any $0 \leq \delta \leq  1$, 
	\[
	\text{Pr}(X \leq (1-\delta)\mu) \leq e^{-\frac{\delta^2\mu}{2}}
	\]
\end{lemma}

%% file: algorithm.tex
\section{LP rounding algorithm for AHTP} 
\label{sec:alg}
In this section, we present our algorithm for the \textsf{AHTP} and prove Theorem~\ref{thm:fractional-tuza}.
Given a $t$-uniform hypergraph $G$ as an input to the $\textsf{AHTP}$, let $H=G^{(t-1)}$ be the $(t-1)$-blown-up hypergraph of $G$. 

\subsection{Color-coding based small vertex cover}

We first prove a lemma that in any $(t-1)$-blown-up hypergraph $H=([n],E)$, there is a vertex cover of size at most $O(\frac{\log t}{t})n$ using a color-coding argument. This lemma illustrates the color-coding idea well, and is also useful later in the context of structural characterization of the blown-up hypergraphs. This lemma is not used in the main algorithm, and the reader can skip to~\Cref{subsec:alg} for the algorithm. 
\begin{lemma}
\label{lem:color-coding}
Suppose $G=([n],E(G))$ is a $t$-uniform hypergraph and $H=G^{(t-1)}=(V(H),E(H))$. Then, there exists a randomized polynomial time algorithm that outputs a vertex cover of $H$ with expected size at most $|V|\left(\frac{2\ln t}{t}+O\left(\frac{1}{t}\right)\right)$. 
\end{lemma}

\begin{proof}
    Our algorithm is based on the color-coding technique used to get upper bounds for the hypergraph Tur\'{a}n problem~\cite{Kim1983, Sidorenko95}.
    Let $P=\left\lceil \frac{t-1}{2\ln t} \right\rceil$. Color each vertex of $G$ with $c:[n]\rightarrow [P]$ uniformly independently at random.
    For $v \in V(H)$ and $i \in [P]$, let $C_i(v)$ denote the number of nodes of $v$ that are colored with $i$, i.e.,
    $C_i(v) := \left| \{j \in v : c(j)=i \} \right|$.
    
    We define a function $f: V(H)\rightarrow \Z_P$ as 
    \[
        f(v) =  C_1(v) + 2 C_2(v) + \ldots + (P-1) C_{(P-1)}(v)  \mod P
    \]
    For an element $i \in \Z_P$, let $f^{-1}(i)$ denote the set $\{ v \in V(H) : f(v)=i\}$. 
    Let $p \in \Z_P$ be such that $|f^{-1}(p)| \leq |f^{-1}(i)|$ for all $i \in \Z_P$. 
    Note that by definition, $|f^{-1}(p)|\leq \frac{|V|}{P}$.
    Let $U \subseteq V(H)$ be defined as follows: 
    \[
    U = \{ v : v \in V(H), \exists i \in [P] \text{ such that }C_i(v)=0 \}
    \]
    
    We claim that $S = f^{-1}(p) \cup U$ is a vertex cover of $H$.
    Consider an arbitrary edge $e = \{ v_1, v_2, \ldots, v_t \} \in E(H)$. Let the corresponding edge in $G$ be equal to $e(G) = \bigcup_{j \in [t]}v_j =(u_1, u_2, \ldots, u_t) \in E(G)$ where $u_1, u_2, \ldots, u_t$ are elements of $[n]$.
    Without loss of generality, let $v_j = e(G) \setminus \{ u_j \} $.
    For a color $i \in [P]$, let $C_i(e) = \left| \{ j \in [t] : c(u_j)=i \} \right| $. 
    We consider two cases separately: 
    \begin{enumerate}
        \item First, if there exists a color $i \in [P]$ such that $C_i(e)=0$, then for every $j \in [t]$, $C_i (v_j)=0$, and thus, for every $j \in [t]$, $v_j\subseteq U$, and thus, $e \cap S \neq \phi.$ 
        \item Suppose that for every color $i \in [P]$, $C_i(e)>0$. We define $f(e)\in \Z_P$ as 
        \[
        f(e) = C_1(e) + 2 C_2(e) + \ldots + (P-1) C_{(P-1)}(e)  \mod P
        \]
        Note that for every $j \in [t]$, we have
        \[
        f(v_j) = f(e)-c(u_j) \mod P
        \]
        As the size of $\{ c(u_1), c(u_2), \ldots, c(u_t) \}$ is equal to $P$, the size of the set $\{ f(v_1), f(v_2), \ldots, f(v_t)\}$ is equal to $P$ as well. Thus, there exists a $j \in [t]$ such that $f(v_j)=p$ which implies that $v_j \in S$. 
    \end{enumerate}
    
    Thus, our goal is to upper bound the expected value of $|S|$. 
    Note that $P \leq \frac{t-1}{\ln t}$.
    By taking union bound over all the colors, we get 
    \[  \E[U] \leq P \left( 1 - \frac{1}{P}\right) ^{t-1}|V| \leq \frac{t-1}{\ln t}e^{-2\ln t}|V|     \leq \left( \frac{1}{t \ln t} \right) |V| \leq O\left( \frac{1}{t} \right) |V|
    \]
    %\begin{align*}
       % \E[U] &\leq P \left( 1 - \frac{1}{P}\right) ^{t-1}|V| \\ 
       % &\leq \frac{t-1}{\ln t}e^{-2\ln t}|V| \\
        %&\leq \left( \frac{1}{t \ln t} \right) |V| = O\left( \frac{1}{t} \right) |V|
    %\end{align*}
    Thus, the expected value of $S$ is at most $|f^{-1}(p)|+\E[|U|]$ which is at most $\left(\frac{2\ln t}{t-1}+O\left(\frac{1}{t}\right)\right)|V|.$ \end{proof}

\subsection{LP rounding based algorithm for \textsf{AHTP}} 

\label{subsec:alg}
Consider the standard LP relaxation for vertex cover in $H$: 
\begin{align*}
\text{Minimize }\sum_{v \in V(H)} &x_v  \\ 
\text{such that}\quad \sum_{ v \in e } x_v & \geq 1 \,\, \forall e \in E(H)  \\ 
x_v &\geq 0 \, \, \forall v \in V(H)
\end{align*}
Let $\overline{x}$ be an optimal solution to the above Linear Program, and let $\textsf{OPT}=\sum_{v \in V(H)}\overline{x}_v$. 
Let $S \subseteq V(H)$ be the set of vertices that are assigned positive LP value i.e.
\[
S = \{ v \in V(H) : \overline{x}_v > 0 \}
\]
We need a lemma relating $|S|$ and $\textsf{OPT}$: 
\begin{lemma}
	\label{lem:dual}
	The cardinality of $S$ is at most $t\cdot \textsf{OPT}$. 
\end{lemma}
\begin{proof}
	Consider the dual of the vertex cover LP: 
	\begin{align*}
	\text{Maximize } \sum_{e \in E(H)} & y(e) \\ 
	\text{such that } \sum_{e \ni v} y(e) &\leq 1 \,\, \forall v \in V(H) \\ 
	y(e) &\geq 0 \, \, \forall e \in E(H)
	\end{align*}
Let $\overline{y}$ be an optimal solution to the above matching LP. By LP-duality, we get $\sum_{e \in E(H)} \overline{y}_e = \textsf{OPT}$. 
Recall that for all $v \in S$, $\overline{x}_v \neq 0$. By the complementary slackness conditions, we get that for all $v \in S$, $\sum_{e \ni v} \overline{y}_e = 1$. Summing over all $v \in S$, we obtain 
\[
|S| = \sum_{v \in S} \sum_{e \ni v } \overline{y}_e \leq t \sum_{e \in E(H)} \overline{y}_e = t\cdot \textsf{OPT}. \qedhere
\] 
\end{proof}

In general, $\textsf{OPT}$ could be much smaller than $|V(H)|$, and thus we cannot use~\Cref{lem:color-coding} directly to obtain algorithm for \textsf{AHTP}.
However, we can obtain a simple $(t-1)$-factor approximation algorithm for \textsf{AHTP} using~\Cref{lem:color-coding}, extending the proof of fractional Tuza's conjecture of Krivelevich~\cite{Kri95}. We consider two different cases: 
\begin{enumerate}
    \item Suppose that there is a vertex $v \in V(H)$ such that $\overline{x}_v =0$. Consider an arbitrary edge $e \in E(H)$ with $v \in e$. As $\sum_{u\in e}\overline{x}_{u} \geq 1$, we can infer that there is a vertex $v' \in e$ such that $\overline{x}_{v'}\geq \frac{1}{t-1}$. We round $v'$ to $1$ i.e. add $v'$ to our vertex cover solution, delete all the edges containing $v'$ and recursively proceed.
    \item Suppose that for every vertex $v \in V(H)$, we have $\overline{x}_v >0$. In this case, using~\Cref{lem:color-coding}, we can find a vertex cover of size $O(\frac{\log t}{t})|V(H)|$, which can be bounded above by $O(\log t)\textsf{OPT}$ using~\Cref{lem:dual}.
\end{enumerate}
We now describe a randomized algorithm to round the LP to obtain an integral solution whose expected size is at most $ \left( \frac{t}{2} +2\sqrt{t\ln t}\right)\textsf{OPT}$. As is evident from the second case in the above $(t-1)$-factor algorithm, the problem is easy when the set of vertices $S \subseteq V(H)$ with non-zero LP value is large. Instead of considering the two different cases based on whether $S=V(H)$ or not, we take a more direct approach by finding a vertex cover of size $\left( \frac 12+o(1)\right)|S|$. Combined with~\Cref{lem:dual}, we get our required approximation guarantee. 

For ease of notation, let $t'=\frac{t}{2} +2\sqrt{t\ln t}$. Our first step is to round all the variables above a certain threshold to $1$ (\Cref{alg:thresholding}). However, we need to do it recursively to ensure that we can bound the optimal value of the remaining instance. 
\begin{algorithm}[!ht]
\caption{Recursive thresholding for \textsf{AHTP}}

\begin{algorithmic}[1]
    \label{alg:thresholding}
    \State Let $\gamma = \frac{1}{t'}$.
    \State Let $\overline{x}$ be an optimal solution of the LP and let $V' = \left\{v:\overline{x}_v \geq \gamma \right\}$.
    \State Let $U=V'$. 
    \While {$V'$ is non-empty}
    \State Delete $V'$ from $V(H)$, and delete all the edges $e \in E(H)$ that contain at least one vertex $v \in V'$. 
    \State Solve the LP with updated $H$. Update $\overline{x}$ to be the new LP solution. \label{step:middle}
    \State Update $V' = \left\{v \in V(H):\overline{x}_v \geq \gamma \right\}$. Update $U \leftarrow U \cup V'$.
    \EndWhile
    \State Output $U$ and the updated $H$.
    \end{algorithmic}
\end{algorithm}

Let the final updated hypergraph $H$ when~\Cref{alg:thresholding} terminates be denoted by $H'$. Let the optimal cost of the solution $\overline{x}$ for the vertex cover on $H'$ be denoted by $\textsf{OPT}'$. 
We prove that the size of the vertex cover output by the algorithm is not too large:
\begin{lemma}
	\label{lem:rounding-high}
	When the above recursive thresholding algorithm (\Cref{alg:thresholding}) terminates, we have $|U| \leq  t'\cdot \left( \textsf{OPT}-\textsf{OPT}'\right)$. 
\end{lemma}
\begin{proof}
	We will inductively prove the following: after line~\ref{step:middle} in the while loop of the algorithm, $|U|\leq  t'\cdot \left( \textsf{OPT} - \textsf{OPT}_{new}\right)$ where $\textsf{OPT}_{new}$ is the cost of the current optimal solution $\overline{x}$. Let $\overline{x}'$ is the optimal solution before deleting $V'$ from $H$. Let $\textsf{OPT}_{old}$ be the cost of the solution $\overline{x}'$. 
	By inductive hypothesis, we have $|U|-|V'| \leq  t'\cdot \left( \textsf{OPT} -\textsf{OPT}_{old}\right)$.  
	
	We claim that $|V'| \leq  t'\cdot \left( \textsf{OPT}_{old} -\textsf{OPT}_{new}\right)$.
	As $\overline{x}$ is an optimal vertex cover of $H$, we have that $\overline{x}'$ restricted to $H$ has cost at least $\textsf{OPT}_{new}$. 
	This implies that $\sum_{v \in V'} \overline{x}_v' \geq \textsf{OPT}_{old}-\textsf{OPT}_{new}$. 
	As each $\overline{x}_v', v \in V'$ is at least $\frac{1}{t'}$, we obtain the required claim. 
\end{proof}

%Henceforth, we assume that each variable is assigned value at most $\frac{1}{\frac{t+1}{2}+\epsilon}$ in the final instance and our goal is to find a vertex cover that has cost at most $\left( \frac{t+1}{2} + \epsilon \right) \textsf{OPT}'$. 

We are now ready to state our main algorithm for the \textsf{AHTP}. The input to the algorithm is a $t$-uniform hypergraph $G$, and the output is a vertex cover for the hypergraph $H=G^{(t-1)}$.
\begin{algorithm}
\caption{Main algorithm}

\begin{algorithmic}[1]
\label{alg:main}
    \State Apply~\Cref{alg:thresholding} to obtain $U$ and let $H'=(V(H'),E(H')) $ be the updated $H$. 
    Let $\overline{x}$ be an optimal solution of the vertex cover LP on $H'$ with $\overline{x}_v \leq \gamma$ for all $v \in V(H')$.
    \State Let $S \subseteq V(H')$ be defined as $S = \{ v : V(H') : \overline{x}_v > 0\}$. 
    \State Let $\delta = \sqrt{\frac{4\ln t}{t-1}}$.
    \State Color the vertices $[n]$ of $G$ using $c:[n]\rightarrow \{0,1\}$ uniformly and independently at random. 
    \State For a vertex $v \in S$ and a color $i \in \{0,1\}$, let $C_i(v)$ denote the number of nodes that are colored with the color $i$ i.e. \Comment{Recall that $S\subseteq V(H')\subseteq \binom{[n]}{t-1}$.} 
    \[
    C_i(v)= | \{ j \in v : c(j)= i \}|
    \] 
    \State Let $S' \subseteq S$ be defined as the set of vertices in $S$ where the discrepancy between two colors is high:
    \[
    S' = \left\{ v \in S : \exists i \in \{0,1\}: C_i(v) \leq (1-\delta)\frac{t-1}{2} \right\} 
    \]
    \State We now define a function $f:  S \rightarrow \{0,1\}$ as $f(v)=C_1(v) \mod 2$.
    \State For  $i \in \{0,1\}$, let $f^{-1}(i)$ denote the set of all the vertices $v \in S$ such that $f(v)=i$.
    \State Let $p \in \{0,1\}$ be such that $|f^{-1}(p)|\leq |f^{-1}(1-p)|$.
    %for all $i \in \{0,1\}$. 
    \State Let $T \subseteq S$ be defined as $T=S' \cup f^{-1}(p)$.
    \State Output $T\cup U$. 
    \end{algorithmic}
\end{algorithm}  

\subsection{Analysis of the algorithm and proof of Theorem~\ref{thm:fractional-tuza}}

We will first prove that~\Cref{alg:main}  indeed outputs a valid vertex cover of $H$. 
\begin{lemma}
\label{lem:vertex-cover}
	$T\cup U$ is a vertex cover of $H$.
\end{lemma}
\begin{proof}
    It suffices to prove that $T$ is a vertex cover of $H'$.
    
	Consider an arbitrary edge $e = (v_1, v_2, \ldots, v_t) \in E(H')$ corresponding to the edge $e(G)= \cup_{j \in [t]}v_j =\{u_1, u_2, \ldots, u_t \} \in E(G)$.  
    Since $\overline{x}_v \leq \gamma$ for all $v \in V(H')$, we can deduce that $ |e \cap S| \geq \frac{1}{\gamma} =t' $.
    
    Our goal is to show that there exists $j \in [t]$ such that $v_j \in T$.
	We consider two separate cases:	\begin{enumerate}
	\item If there is a color $i \in \{0,1\}$ such that there are at most $(1-\delta)\frac{t-1}{2}$ nodes of color $i$ in $e(G)$, then for all $j \in [t], C_i(v_j) \leq (1-\delta)\frac{t-1}{2}$.
	Since $e \cap S$ is non-empty, there exists $j \in [t]$ such that $v_j \in S$. By definition of $S'$, this implies that $v_j \in S'$ as well, and thus $e \cap T \neq \phi$.  

	\item Suppose that in the coloring $c$, both the colors $0,1$ occur at least $(1-\delta)\frac{t-1}{2}$ times in $e$. 
	Let $e' = e \cap S$ and let $k=|e'| \geq t'$. Without loss of generality, let $e' = \{v_1, v_2, \ldots, v_{k}\}$. For every $j \in [k]$, let $v_j = e(G) \setminus \{ u_j\}$ for $u_j \in [n]$. 
    First, we claim that $t-k < (1-\delta)\frac{t-1}{2}$. We have 
    \begin{align*}
        t-k - (1-\delta)\frac{t-1}{2} &\leq t-t'-(1-\delta)\frac{t-1}{2} \\ 
        &= \frac{t}{2}-2\sqrt{t\ln t} - \left(1-\sqrt{\frac{4\ln t}{t-1}}\right) \frac{t-1}{2} \\ 
        &= \frac{1}{2} \left( t - 4\sqrt{t\ln t} -(t-1)+2\sqrt{(t-1)\ln t} \right) \\
        &\leq \frac{1}{2} \left( 1 - 2 \sqrt{t \ln t} \right) < 0
    \end{align*}
	Since each color occurs at least $(1-\delta) \frac{t-1}{2}$ times in $e(G)$, using the above, we can infer that 
	\[
	\left|\{c(u_1), c(u_2), \ldots, c(u_k)\}\right| \geq 2.
	\]
	We define the value $f(e)$ in the same fashion as we have defined $f(v)$ for $v\in S$: For $i \in \{0,1\}$, let $C_i(e)$ denote the number of nodes $j\in [t]$ such that $c(u_j)=i$, and let $f(e)=C_1(e) \mod 2$. 
	 Using this definition, we get 
	 \[ 
	 f(v_j) = f(e)-c(u_j) \mod 2 \,\,\forall j \in [k]. 
	 \]
	 As $\{c(u_1), c(u_2), \ldots, c(u_k)\}=\{0,1\}$, we have
	 $\{ f(v_1), f(v_2), \ldots, f(v_k)\}=\{0,1\}$ as well. Thus, there exists $j \in [k]$ such that $f(v_j)=p$, which proves that $v_j \in f^{-1}(p) \subseteq T$.\qedhere 
	\end{enumerate}  
\end{proof}

Note that the expected number of nodes of each color $i \in \{0,1\}$ in a vertex $v=(u_1, u_2, \ldots, u_{t-1}) \in S$ is equal to $\frac{t-1}{2}$. 
The set $S'$ is the set of vertices of $S$ where there is a color that occurs much fewer than its expected value. 
We prove that this happens with low probability:
\begin{lemma}
\label{lem:expected}
	The expected cardinality of $S'$ is at most $\frac{2}{t}|S|$.
\end{lemma}
\begin{proof}
	Let $v=(u_1,u_2, \ldots, u_{t-1}) \in S$ be an arbitrary vertex in $S$, where $u_1, u_2, \ldots, u_{t-1}$ are elements of $[n]$. For a color $i \in \{0,1\}$, let the random variable $X(i)$ denote to the number of nodes $j \in [t-1]$ such that $c(u_j)=i$. We can write $X(i) = \sum_{j \in [t-1]} X(i,j)$, where $X(i,j)$ is the indicator random variable of the event that $c(u_j)=i$.  
	We have $\mu = \E[X(i)]=\frac{t-1}{2}$. 
	Using multiplicative Chernoff bound (\Cref{thm:chernoff}), we can upper bound the probability that $X(i) \leq (1-\delta)\frac{t-1}{2}$ by 
	\[
	\text{Pr}\left(X(i) \leq (1-\delta)\frac{t-1}{2}\right) \leq e^{-\frac{\delta^2(t-1)}{4} }
	\] 
For the choice $\delta = \sqrt{\frac{4\ln t}{t-1}}$, the above probability is at most $\frac{1}{t}$. By applying union bound over the two colors and adding the expectation over all the vertices in $S$, we obtain the lemma. 
\end{proof}

Finally, we bound the expected size of the output of the algorithm: 
\begin{lemma}
\label{lem:expectation-output}
    The expected cardinality of $T\cup U$ is at most $ \left(\frac{t}{2}+2\sqrt{t\ln t}\right) \cdot \textsf{OPT}$. 
\end{lemma}
\begin{proof}
Note that by definition, $|f^{-1}(p)|\leq \frac{|S|}{2}$. We bound the expected size of the output of the algorithm $T\cup U$ as   
\begin{align*}
    \E[|T\cup U|] &\leq \E[|T|] + \E[|U|] 
    \ \leq \ \E[|S'|] + \frac{1}{2}|S|+\E[|U|] \\ 
    &\leq \left( \frac{1}{2}+ \frac{2}{t}\right)|S| + \E[|U|] \quad(\text{Using~\Cref{lem:expected}}) \\ 
    &\leq \left( \frac{t}{2} +2 \right) \textsf{OPT'} + \E[|U|] 
    \quad(\text{Using~\Cref{lem:dual}}) \\
    &\leq \left(\frac{t}{2}+2\sqrt{t\ln t}\right) \textsf{OPT}
        \quad(\text{Using~\Cref{lem:rounding-high}}) \ . \qedhere
\end{align*}
\end{proof}

\noindent 
\Cref{lem:vertex-cover} and~\Cref{lem:expectation-output} together imply~\Cref{thm:fractional-tuza}.

\subsection{$(t,2)$-version of \textsf{AHTP}}
\label{subset-(t,2)}
An interesting generalization of \textsf{AHTP} is the $(t,k)$-version, the problem of vertex cover on the $k$-blown-up hypergraph $H=G^{(k)}$ for a $t$-uniform hypergraph $G$, for an arbitrary $1\leq k < t$. 
The case of $k=1$ is the standard vertex cover on $t$-uniform hypergraphs, and $k=t-1$ is the \textsf{AHTP}. Note that there is a trivial $\binom{t}{k}$-factor approximation algorithm for this problem as it can be cast as an instance of vertex cover on a $\binom{t}{k}$-uniform hypergraph. 
The above algorithm can be shown to achieve a $\binom{t}{k}c(k)$ approximation guarantee for the general problem where $c(k)\rightarrow \frac 12+o(1)$ as $k \rightarrow t-1$. 

We now turn our attention to the interesting case of $k=2$.
When the hypergraph $G$ consists of $t$-cliques in a graph, the vertex cover problem on $G^{(2)}$ is the generalization of Tuza's problem where we try to hit all $t$-cliques with the fewest possible edges. 
 Note that in this case, the trivial hypergraph vertex cover algorithm achieves a $\binom{t}{2}$-factor approximation. We describe how a simplified version of our algorithm can be used to get a $\frac{t^2}{4}$-factor guarantee: Let $H=G^{(2)}=(V(H),E(H))$, and we iteratively solve the Vertex Cover LP on $H$ to round all the vertices with value at least $\frac{4}{t^2}$. In the remaining instance, we let $S \subseteq V(H)$ to be the vertices of $H$ that are assigned non-zero LP value i.e. $S= \{ v \in V(H) : x_v > 0 \}$. We use a color coding function $c:[n]\rightarrow \{0,1\}$ picked uniformly and independently at random, and we output all the vertices $T=\{ \{i,j\} \in S: c(i)=c(j)\}$. The expected size of $T$ is at most $\frac{1}{2}\binom{t}{2}\textsf{OPT}\leq \frac{t^2}{4}\textsf{OPT}$ as the cardinality of $S$ is at most $\binom{t}{2}\textsf{OPT}$. 
 
We now argue that $T$ is indeed a vertex cover of $H$. Consider an edge $e=\{v_{i,j}:i\neq j \in [t]\}\in E(H)$ corresponding to the edge $e'=(u_1, u_2, \ldots, u_t)\in E(G)$.
Recall that every element of $e$ corresponds to a subset of size $2$ of $e'$, and thus, without loss of generality, let $v_{i,j}=\{u_i, u_j\}$ for all $i,j \in [t]$.
As $x_{v_{i,j}}< \frac{4}{t^2}$ for all $i,j\in [t]$, there are greater than $\frac{t^2}{4}$ pairs of indices $i,j$ such that $x_{v_{i,j}}>0$, or equivalently, $v_{i,j} \in S$. 
Thus, $|e \cap S| > \frac{t^2}{4}$.
For every function $c:[t]\rightarrow \{0,1\}$, the number of pairs of indices $i \neq j \in [t]$ such that $c(i)\neq c(j)$ is at most $\frac{t^2}{4}$. Thus, there are at least $\binom{t}{2}-\frac{t^2}{4}$ pairs of indices $i\neq j \in [t]$ such that $c(u_i)=c(u_j)$. As $|e \cap S|>\frac{t^2}{4}$, there exists a pair of indices $i\neq j \in [t]$ such that $v_{i,j} \in S$, $c(u_i)=c(u_j)$, which implies that $v_{i,j} \in T$. Thus, for every edge $e \in E(H)$ of $H$, there exists an element $v \in e$ such that $v\in T$, which completes the proof that $T$ is a vertex cover of $H$. 

As a corollary of the algorithm, we deduce that for all hypergraphs $H=G^{(2)}$ of a $t$-uniform hypergraph $G$,
\[
\tau(H) \leq \frac{t^2}{4} \tau^*(H)
\]
This proves the fractional version of a conjecture due to Aharoni and Zerbib~\cite{AZ20} (Conjecture $1.4$) for the case when $k=2$.

%% file: structure.tex
\section{Forbidden sub-hypergraphs and Tuza's conjecture}
\label{sec:structure}
   Since \textsf{AHTP} is the problem of vertex cover on $H=G^{(t-1)}$ for a given $t$-uniform hypergraph $G$, an interesting question is to characterize the $t$-uniform hypergraphs $H$ that can arise as the blown-up hypergraph $G^{(t-1)}$ of some $t$-uniform hypergraph $G$. A very simple necessary condition is that the hypergraph $H$ should be simple. However, this is not sufficient---there are simple $t$-uniform hypergraphs $H$ that cannot be written as $H=G^{(t-1)}$ for any $G$. For example, the $t$-tent hypergraph (\Cref{def:tent}) is a simple hypergraph, but cannot be written as a $(t-1)$-blown-up hypergraph. 
   A natural question in this context is the following: 
   \begin{quote}
       Is there a finite set of hypergraphs $\mathcal{F}$ such that every hypergraph that does not have any member of $\mathcal{F}$ as a  sub-hypergraph can be represented as $G^{(t-1)}$ for some $t$-uniform hypergraph $G$?
   \end{quote}
In addition to its inherent structural interest, the above question can shed light on Tuza's conjecture.
Recall that Aharoni and Zerbib\cite{AZ20} proposed a generalization of Tuza's conjecture stating that $\tau\left(G^{(2)}\right)\leq 2\cdot \nu\left(G^{(2)}\right)$ for all $3$-uniform hypergraphs $G$. They suggested that understanding the structure of blown-up hypergraphs, and specifically, the sub-hypergraphs that it excludes might be a promising approach to establish this conjecture. In particular, they observed that the blown-up hypergraphs do not contain ``tents" as a sub-hypergraph. 
\begin{definition}
	\label{def:tent}
	A $t$-tent (\Cref{fig:tent}) is a set of four $t$-uniform edges $e_1, e_2, e_3, e_4$ such that 
	\begin{enumerate}
		\item $\cap_{i=1}^3 e_i \neq \phi $. 
		\item $| e_4 \cap e_i | = 1 $ for all $i \in [3]$. 
		\item $ e_4 \cap e_i \neq e_4 \cap e_j $ for all $ i \neq j \in [3]$. 
	\end{enumerate}
\end{definition}
In~\cite{AZ20}, the authors pose the following question. Note that an answer in the affirmative would resolve Tuza's conjecture, and in fact its above generalization that  $\tau\left(G^{(2)}\right)\leq 2\nu\left(G^{(2)}\right)$ for all $3$-uniform hypergraphs $G$. 
\begin{problem}
\label{prob:tents}
	Is it true that for every $3$-uniform hypergraph $H$ without a $3$-tent, $\tau(H) \leq 2 \cdot \nu(H)$?
\end{problem}
We answer this question in the negative. In fact, we prove a stronger statement that there can be no forbidden substructure-based Tuza's theorem. 
%\vg{Is there any relation to Sherali-Adams gaps here? Might be good to say something in this regard, just so some reviewer doesn't wonder or feel that this might follow from known SA gaps}

\begin{theorem}
\label{thm:structure}
    Let $\mathcal{F}=\{F_1, F_2, \ldots, F_{\ell}\}$ be an arbitrary set of  $t$-uniform hypergraphs such that for every $t$-uniform hypergraph $G$, the blown-up hypergraph $G^{(t-1)}$ does not contain any $F_i \in \mathcal{F}$ as a sub-hypergraph.
    
    Then, for every $\epsilon >0$, there exists a hypergraph $H'$ that does not contain any member of $\mathcal{F}$ as  a sub-hypergraph and which satisfies  $\tau(H')\geq (t-\epsilon)\tau^*(H') \geq (t-\epsilon)\nu(H')$.
\end{theorem}

By setting $\mathcal{F}$ to be the single $3$-tent hypergraph, we obtain a counterexample to~\Cref{prob:tents}.
Furthermore, when $t=3$, the construction we give to prove Theorem~\ref{thm:structure} will belong to the 
class of $3$-uniform hypergraphs $H$ obtained from a given graph $G$ with the vertex set of $H$ being the edge set of $G$, and every triangle in $G$ forming an edge in $H$. Thus, there is no ``local'' proof of Tuza's conjecture that uses only substructure properties of the underlying hypergraph. 

We call a hypergraph non-trivial if it has at least two edges. Before we prove the above theorem, we use a definition from~\cite{FM08}. 
\begin{definition}
    Let $F$ be a non-trivial $t$-uniform hypergraph. Then, 
    \[
    \rho(F)=\max_{F' \subseteq F}\frac{e'-1}{v'-t}
    \]
    where $F'$ is a non-trivial subhypergraph of $F$ with $e'>1$ edges and $v'$ vertices.  
\end{definition}

We now return to the proof of Theorem~\ref{thm:structure}.
%TODO define triangle-edge hypergraph. 
\begin{proof}(of Theorem~\ref{thm:structure})
We will first prove that $\rho(F_i)>\frac{1}{t-1}$ for all $i \in [\ell]$. Suppose for contradiction that there exists a $t$-uniform hypergraph $F_i \in \mathcal{F}$ such that $\rho(F_i)\leq \frac{1}{t-1}$. 
Without loss of generality, we can assume that $F_i$ is connected. 
Order the edges of $F_i$ as $\{e_1, e_2, \ldots, e_m\}$ such that for every $j>1$, $e_j \cap (e_1 \cup e_2 \cup \ldots \cup e_{j-1}) \neq \phi.$ For every $j \geq 1, $ let $F'_j$ be the subhypergraph induced by $\{e_1, e_2, \ldots, e_j\}$. As $\rho(F_i)\leq \frac{1}{t-1}$, we can infer that for every $j>1$,
\[
\frac{|E(F'_j)|-1}{|V(F'_j)|-t} \leq \frac{1}{t-1}
\]
which implies that $|V(F'_j)|\geq (t-1)|E(F'_j)|+1 =(t-1)j+1$. 
As $|V(F'_j)|= |V(F'_{j-1})|+t-|e_j \cap ( e_1 \cup e_2 \cup \ldots \cup e_{j-1})|$, we get that $|V(F'_j)|\leq |V(F'_{j-1})|+t-1$, which combined with the above shows that the inequality is in fact tight for every $j >1$. Thus, for every $j > 1, |V(F'_j)|=|V(F'_{j-1})|+t-1$, which implies that for every $j >1$,
\[
|e_j \cap ( e_1 \cup e_2 \cup \ldots \cup e_{j-1})|=1.
\]

We now construct a $t$-uniform hypergraph $H$ such that $F_i$ is isomorphic to $H^{(t-1)}$. We construct the hypergraph $H$ inductively via $H_1, H_2, \ldots, H_m=H$ such that $H_j^{(t-1)}$ is isomorphic to $F'_j$ for all $j \in [m]$. First, we set the hypergraph $H_1$ to be equal to the $t$-uniform hypergraph on $t$ vertices with a single edge. $H_1$ is trivially isomorphic to $F'_1$. Assume by inductive hypothesis that there is a hypergraph $H_k$ such that $F'_k$ is isomorphic to $H_k^{(t-1)}$ for some $k \in [m-1]$. Let $\phi : F'_k \rightarrow H_k^{(t-1)}$ be the isomorphism between the two hypergraphs. 
The hypergraph $F'_{k+1}$ is obtained from $F'_k$ by adding an edge $e_{k+1}$ such that $e_{k+1}$ intersects with $F'_k$ in exactly one vertex $v\in V(F'_k)$. Recall that the vertex set of $H_k^{(t-1)}$ is the set of subsets of vertices of $H_k$ of size $t-1$. Thus, $\phi(v)=\{(p_1,p_2, \ldots, p_{t-1})\}$ for a set of vertices $p_1,p_2, \ldots, p_{t-1} \in V(H_k)$. We construct $H_{k+1}$ by introducing a new vertex $v'$ and adding the edge $\{v',p_1,p_2, \ldots, p_{t-1}\}$ to the hypergraph $H_k$. Thus, $H_{k+1}^{(t-1)}$ is obtained from $H_{k}^{(t-1)}$ by adding single edge that intersects with $H_{k}^{(t-1)}$ at exactly one vertex, that is $\{p_1, p_2, \ldots, p_{t-1} \}$. Hence, $H_{k+1}^{(t-1)}$ is isomorphic to $F'_{k+1}$, completing the proof. 
This proves that there exists a $t$-uniform hypergraph $H=H_m$ such that $F_i=H^{(t-1)}$, contradicting the fact that no $(t-1)$-blown-up hypergraph contains $F_i$ as a subhypergraph. 

Thus, $\rho(F_i)>\frac{1}{t-1}$ for all $i \in [\ell]$. Let $\rho = \min_{i \in [\ell]}\rho(F_i) > \frac{1}{t-1}$.
Consider a random $t$-uniform hypergraph $H'$ on $n$ vertices sampled by picking each edge independently with probability $p=n^{-\frac{1}{\rho}}$. We now delete the edges in a maximal collection of edge disjoint copies of members of $\mathcal{F}$ from $H'$. It has been proved~\cite{BFM10,FM08} that the maximum independent set $\alpha(H')$ of this construction satisfies 
\[
\alpha(H') \leq \Tilde{O}\left(n^{\frac{1}{(t-1)\rho}} \right) 
\]
with high probability. Thus, there exists a $t$-uniform hypergraph $H'$ with $n$ vertices without any substructure from $\mathcal{F}$ such that $\tau(H')\ge (1-o(1))n$. Since for any $t$-uniform hypergraph $H'$ on $n$ vertices, $\nu(H')\leq \tau^*(H')\leq \frac{1}{t}n$, this proves the claimed factor $(t-\epsilon)$ gap between $\tau(H')$ and $\tau^*(H')$ for every positive constant $\epsilon>0$.
\end{proof}

\subsection{Explicit construction of tent-free hypergraphs}
We now describe an explicit hypergraph giving negative answer to~\Cref{prob:tents}. Our counterexample is a hypergraph with vertex set $[3]^n$ for large enough $n$ and the edge set is the set of all combinatorial lines that we formally define below: 
\begin{definition}(Combinatorial lines in $[3]^n$)
    A set of three distinct vectors $u=(u_1, u_2, \ldots, u_n), v=(v_1, v_2, \ldots, v_n), w=(w_1, w_2, \ldots, w_n) \in [3]^n$ forms a combinatorial line if there exists a subset $S \subseteq [n]$ such that 
	\begin{enumerate}
		\item For all $i \in [n]\setminus S$, $u_i = v_i = w_i$. 
		\item There exist three distinct integers $u', v', w' \in [3]$ such that for all $i \in S$, $u_i = u', v_i = v', w_i = w'$. 
	\end{enumerate}
\end{definition}
We will use the following seminal result about combinatorial lines: 
\begin{theorem}(Density Hales Jewett Theorem \cite{FK91},\cite{poly12} )
	\label{thm:dhj}
	For every positive integer $k$ and every real number $\delta > 0$ there exists a positive integer $\textsf{DHJ}(k,\delta)$ such that if $n \geq \textsf{DHJ}(k, \delta)$ and $A$ is any subset of $[k]^n$ of density at least $\delta$, then $A$ contains a combinatorial line. 
\end{theorem}
We now prove~\Cref{thm:tent}.
\begin{T1}
  	For every $\epsilon>0$, there exists a $3$-uniform hypergraph $H$ without a $3$-tent such that $\tau(H) > (3-\epsilon) \nu(H)$. 
\end{T1}
\begin{proof}
%	We will prove that there exist $3$-uniform hypergraphs on $N$ vertices with the minimum vertex cover $\tau = (1- o(1))N$. 
	
	The hypergraph that we use $H=(V,E)$ has $V=[3]^n$ for $n$ large enough to be set later, and the edges are all the combinatorial lines in $[3]^n$.  
	First, we claim that the above defined hypergraph does not have a $3$-tent. Suppose for contradiction that there are edges $e_1, e_2, e_3, e_4$ satisfying the properties of~\Cref{def:tent}. Let $u = (u_1, u_2, \ldots, u_n) \in e_4 \cap e_1, v = (v_1, v_2, \ldots, v_n) \in e_4 \in e_2, w = (w_1, w_2, \ldots, w_n) \in e_4 \cap e_3$. Note that $e_4 = \{ u, v, w\}$. Thus, there exists a subset $S \subseteq [n]$ such that for all $i \in [n] \setminus S$, $u_i = v_i = w_i$. Without loss of generality, we can also assume that for all $i \in S$, $u_i = 1, v_i = 2, w_i = 3$. 

% \InsertBoxR{2}{\begin{minipage}{0.45\linewidth}\centering
%         \begin{tikzpicture}[scale=0.4]
%         \draw (0,0)  node[fill,circle, draw, fill=black,inner sep=0,minimum size=0.15cm,label={\small $\quad\,\,  x$}] {};
%         \draw (0,-2) node[fill,circle, draw, fill=black,inner sep=0,minimum size=0.15cm] {};
%         \draw (0,-4) node[fill,circle, draw, fill=black,inner sep=0,minimum size=0.15cm,label={\small $\quad\,\,  v$}] {};
%         \draw (-4,-4) node[fill,circle, draw, fill=black,inner sep=0,minimum size=0.15cm,label={\small $\quad\,\,  u$}] {};
%         \draw (4,-4) node[fill,circle, draw, fill=black,inner sep=0,minimum size=0.15cm,label={\small $\quad\,\,  w$}] {};
%         \draw (-2,-2) node[fill,circle, draw, fill=black,inner sep=0,minimum size=0.15cm] {};
%         \draw (2,-2) node[fill,circle, draw, fill=black,inner sep=0,minimum size=0.15cm] {};
        
%         \draw[rotate=45] (-3,0)ellipse (105pt and 17pt);
%         \draw[rotate=135] (-3,0)ellipse (105pt and 17pt);
%         \draw (0,-4)ellipse (135pt and 17pt);
%         \draw (0,-2)ellipse (17pt and 80pt);
%         \end{tikzpicture}
%             \label{fig:tent-example}
%         \end{minipage}%
%                   }[5]
%\vg{Is it possible to draw some figure here? I didn't read this now, but I recall it was a bit intense when I read it earlier.}	
	Let $x=(x_1, x_2, \ldots, x_n) \in e_1 \cap e_2 \cap e_3$. Note that $\{ x, u\} \subseteq e_1, \{ x, v\} \subseteq e_2, \{ x, w\} \subseteq e_3$. Consider an arbitrary element $p \in S$, and without loss of generality, let $x_p = 1$. Thus, we have that $x_p = 1, v_p = 2$ and both $x, v$ share the combinatorial line $e_2$.  
	This implies that there exist a subset $S_2 \subseteq [n]$ such that for all $i \in [n]\setminus S_2$, $x_i = v_i$ and for all $i \in S_2$, $x_i = 1, v_i = 2$. Similarly, there exists a subset $S_3 \subseteq [n]$ such that for all $i \in [n] \setminus S_3$, $x_i = w_i$ and for all $i \in S_3$, $x_i = 1, w_i = 3$. 
	
	Note that $S_2 \subseteq S$. Suppose for contradiction that there exists $j \in S_2 \setminus S$. Then, we have $v_j = 2, x_j = 1$. However, since $v_i = w_i$ for all $i \in [n] \setminus S$, we get that $w_j = 2$, and thus, $j \notin S_3$, which implies that $x_j = w_j = 2$, a contradiction. 
	Thus, $S_2 \subseteq S$, and similarly $S_3 \subseteq S$.
	We can also observe that $S_2 \neq S$ since in that case, $x = u$ which cannot happen since $|e_4 \cap e_2|=1$. By the same argument on $e_3$, we can deduce that $S_3 \neq S$. 
	As $S_2$ is a strict subset of $S$, there exists $j \in S \setminus S_2$. As $v_i = x_i$ for all $i \in [n]\setminus S_2$, $x_j = v_j = 2$. As $j \in S$, we have $w_j = 3$. However, as $w_j \neq x_j$, this implies that $j \in S_3$, which then implies that $x_j = 1$, a contradiction.    
	
	Now, we will prove that for large enough $n$, $\tau(H) > (3-\epsilon) \nu(H)$. Let $N= 3^n$. Since the cardinality of $V$ is equal to $N$, we have $\nu(H) \leq \frac{N}{3}$. We apply~\Cref{thm:dhj} with $k=3, \delta = \frac{\epsilon}{3}$, and set $n \geq \textsf{DHJ}(k,\delta)$. Thus, we can infer that in any subset $T \subseteq V$ of size $\frac{\epsilon}{3} N$, there exists an edge of $H$ fully contained in $T$. Thus, we get that $\tau(H)> (1-\frac{\epsilon}{3}) N$, which gives $\tau(H)>(3-\epsilon)\nu(H)$. 
\end{proof}

%% file: simple.tex
\section{Vertex cover and set cover on simple hypergraphs}
\label{sec:simple}
As mentioned earlier, the edges in a $(t-1)$-blown-up hypergraph of a $t$-uniform hypergraph can intersect on at most one element, so such hypergraphs are simple.
In this section, we will take a step back and address to what extent improved approximation algorithms are possible for vertex cover on simple hypergraphs. We will also consider the dual problem, of covering the vertices by the fewest possible hyperedges, namely the set cover problem, on simple hypergraphs but without any restriction on the size of the hyperedges. Note that a hypergraph is simple if and only if the edge-vertex incidence bipartite graph does not contain a copy of $K_{2,2}$. 
Thus, a hypergraph is simple if and only if its dual is simple.

\subsection{Vertex cover on simple $t$-uniform hypergraphs}

We now prove~\Cref{thm:simple} which shows that simple hypergraphs are still rich enough to preclude a non-trivial approximation to vertex cover.
Our hardness is established using a reduction from the general problem of vertex cover on $t$-uniform hypergraphs. In particular, we use the following result:
\begin{theorem}(\cite{DGKR05})
	\label{thm:hypergraph-vc}
	For every constant $\epsilon > 0$ and $t \geq 3$, the following holds: Given a $t$-uniform hypergraph $G=(V,E)$, it is NP-hard to distinguish between the following cases: 
	\begin{enumerate}
		\item Completeness: $G$ has a vertex cover of measure $\frac{1+\epsilon}{t-1}$.
		\item Soundness: Any subset of $V$ of measure $\epsilon$ contains an edge from $E$. 
	\end{enumerate}
\end{theorem}

We give a randomized reduction from~\Cref{thm:hypergraph-vc} to~\Cref{thm:simple}. The approach is similar to the one used in \cite{GL-sidma} for showing the inapproximability of $H$-Transversal in graphs. It was also used in the recent tight hardness for Max Coverage on simple set systems~\cite{CKL20a}.

Let us instantiate~\Cref{thm:hypergraph-vc} with $\epsilon$ replaced by $\epsilon' = \frac{\epsilon}{4}$, and let the resulting hypergraph be denoted by $G$. 
Now, given this $t$-uniform hypergraph $G=(V,E)$, we output a $t$-uniform hypergraph $H=(V',E')$ as follows: Let $n=|V|, m=|E|$.
We have integer parameters $B,P$ depending on $\epsilon,t,n,m$ to be set later. 
The vertex set of $H$ is $V'=V \times [B]$--we have a cloud of $B$ vertices $v^1,v^2,\ldots, v^B$ in $V'$ corresponding to every vertex $v \in V$. 
For every edge $e=(v_1, v_2, \ldots, v_t) \in E$, we pick $P$ edges $e^1, e^2,\ldots, e^P$ with $e^i= ((v_1)_i, (v_2)_i, \ldots, (v_t)_i)$ and add them to $E'$, where for each $j \in [t]$ and $i \in [P]$, $(v_j)_i$ is chosen uniformly and independently at random from $(v_j)^1, (v_j)^2, \ldots ,(v_j)^B$. 
Thus, so far, we have added $mP$ edges to $E'$.

We first upper bound the expected value of the number of pairs of edges in $E'$ that intersect in more than one vertex. 
Order the edges in $E'$ as $e_1, e_2, \ldots, e_{mP}$.
Let $X$ denote the random variable that counts the number of pairs of edges in $E'$ that intersect in more than one vertex. 
For every pair of indices $i,j \in [mP]$, let the random variable $X_{ij}$ be the indicator variable of the event that the edges $e_i$ and $e_j$ of $E'$ intersect in greater than one vertex. 
Note that the edges in $E$ corresponding to $e_i$ and $e_j$ have at most $t$ vertices in common. Thus, the probability that $e_i$ and $e_j$ intersect in at least two vertices is upper bounded by $\binom{t}{2}\frac{1}{B^2}$. Summing over all the pairs $i,j$, we get
\[
\E[X] \leq \binom{mP}{2} \binom{t}{2}\frac{1}{B^2} \leq \frac{m^2t^2P^2}{B^2}. 
\]
By Markov's inequality, with probability at least $\frac{9}{10}$, $X$ is at most $\frac{10m^2t^2P^2}{B^2}$.

We consider all the pairs of edges that intersect in more than one vertex in $E'$, and arbitrarily delete one of those edges. Let the resulting set of edges be denoted by $E^{''}$. The final hypergraph resulting in this reduction is $H=(V',E^{''})$.
Note that $H$ is indeed a simple hypergraph. 
We will prove the following: 
\begin{enumerate}
    \item (Completeness) If $G$ has a vertex cover of measure $\mu$, then there is a vertex cover of measure $\mu$ in $H$. 
    \item (Soundness) If every subset of $V$ of measure $\epsilon'$ contains an edge from $E$, then with probability at least $\frac{4}{5}$, every subset of $V'$ of measure $\epsilon$ contains an edge from $E^{''}$. 
\end{enumerate}
\paragraph{Completeness.} If $G$ has a vertex cover of size $\mu n$, then picking all the vertices in $V'$ in the cloud corresponding to these vertices ensures that $H$ has a vertex cover of size $\mu nB$. Thus, in the completeness case, there is a vertex cover of measure $\mu$ in $H$. 

\paragraph{Soundness.} Suppose that every $\epsilon'$ measure subset of $V$ contains an edge from $E$. Our goal is to show that with probability at least $\frac{4}{5}$, every $\epsilon$ measure subset of $V'$ contains an edge from $E^{''}$. 
We first prove the following lemma: 
\begin{lemma}
\label{lem:edge}
With probability at least $\frac{9}{10}$ over the choice of $E'$, the following holds: For every edge $e=(v_1,v_2, \ldots,v_t) \in E$, and every subset $S\subseteq V'$ such that for each $i \in [t]$, $S$ contains at least $\frac{\epsilon}{4}B$ vertices from $\{ v_i^1, v_i^2, \ldots, v_i^B\}$, there exists an edge $e' \in E'$ all of whose vertices are in $S$.
\end{lemma}
\begin{proof}
The probability that there exists an edge $e = (v_1, v_2, \ldots, v_t)$ and a subset $S$ which contains at least $\frac{\epsilon}{4}B$ vertices from each cloud and does not contain any edge from $E^{'}$ is at most \[ 
m 2^{tB} \left(1-\left(\frac{\epsilon}{4}\right)^t\right)^P \leq m 2 ^{tB- \log e \frac{\epsilon^t P}{4^t}} \leq \frac{1}{10}
\]
when $P=maB$ where $a:=a(t,\epsilon) = \frac{4^{t+2}t}{\epsilon^t}$. 
\end{proof}

Using the above lemma, we can conclude that with probability at least $\frac{4}{5}$, $X \leq \frac{10m^2t^2P^2}{B^2}= 10m^4t^2a^2$ and for every edge $e \in E$ and every subset $S\subseteq V'$ such that for each $i \in [t]$, $S$ contains at least $\frac{\epsilon}{4}B$ vertices from $\{ v_i^1, v_i^2, \ldots, v_i^B\}$, there exists an edge $e' \in E'$ all of whose vertices are in $S$. 
We claim that this implies that with probability at least $\frac{4}{5}$, every $\epsilon$ measure subset of $V'$ contains an edge of $E^{''}$. 
Consider an arbitrary subset $U \subseteq V'$ such that $|U| \geq \epsilon n B$. We choose $B$ large enough such that $t(10m^4t^2a^2) \leq \frac{\epsilon}{2}nB$.
Thus, the set of the vertices $W$ in the edges deleted from $E'$ to obtain $E^{''}$ has cardinality at most $\frac{\epsilon}{2}nB$. 

Let $U' = U \setminus W$. Note that all the edges in $U'$ that are in $E'$ are present in $E^{''}$ as well. 
As $U'$ has a measure of at least $\frac{\epsilon}{2}$ in $V'$, for at least $\frac{\epsilon}{4}n$ vertices $v$ in $V$, $U'$ should contain at least $\frac{\epsilon}{4}$ fraction of the vertices in the cloud $\{v^1,v^2,\ldots,v^B\}$. Since otherwise, the cardinality of $U'$ is at most $\left(n-\frac{\epsilon n}{4}\right)\cdot \frac{\epsilon B}{4}+\frac{\epsilon n}{4}\cdot B<\frac{\epsilon nB}{2}$, a contradiction. 
By~\Cref{lem:edge}, we can deduce that there exists an edge $e \in E'$ all of whose vertices are in $U'$, which implies that the edge $e$ is in $E^{''}$ as well. 
This proves that in the soundness case, with probability at least $\frac{4}{5}$, there exists an edge in every $\epsilon$ measure subset of $V'$. 

This completes the proof of~\Cref{thm:simple}. Under the Unique Games Conjecture~\cite{Khot02}, the hardness of vertex cover in $t$-uniform hypergraphs can be improved to $t-\epsilon$. We remark that we can get the same hardness for simple hypergraphs by our reduction. 

\subsection{Set Cover on Simple Set Systems}

In the set cover problem, there is a set family $\mathcal{S}\subseteq 2^X$ on a universe $X=[n]$, and the goal is to cover the universe $[n]$ with as few sets from the family as possible. 
The greedy algorithm where we repeatedly pick the set that covers the maximum number of new elements achieves a $\ln n$-factor approximation algorithm for the problem, and this is known to be optimal. We consider the same problem under the restriction that the family $\mathcal{S}$ is a simple set system i.e. for every $i \neq j, |S_i \cap S_j | \leq 1$. 
Surprisingly, in contrast with the hardness result for vertex cover, simplicity of the set family helps in achieving better approximation factor for the set cover problem.
\begin{theorem}(Theorem~\ref{thm:simple-set-cover-intro} restated)
\label{thm:simple-set-cover}
    The greedy algorithm achieves a $\left(\frac{\ln n}{2}+1\right)$-approximation guarantee for the set cover problem on simple set systems over a universe of size $n$. Furthermore, the bound is essentially tight for the greedy algorithm---there is a simple set system on which the approximation factor of greedy exceeds $(\frac{\ln n}{2}- 1)$.
\end{theorem}
\begin{proof}
First, we prove the upper bound. Let the optimal solution size be equal to $k$ i.e. there is $\mathcal{T}=\{S_1, S_2, \ldots, S_k\} \subseteq \mathcal{S}$ such that the union of sets in $\mathcal{T}$ is equal to $[n]$. For every set $S \in \mathcal{S}\setminus \mathcal{T}$, $|S\cap S_i| \leq 1$ by the simplicity of the set system, and thus, we get that 
\begin{equation}
\label{eq:set-size-bound}
\forall S \in \mathcal{S}\setminus \mathcal{T}, |S|\leq k.
\end{equation}
We now consider two different cases: 
\begin{enumerate}
    \item Suppose that $k \geq \sqrt{n}$. We recall that the greedy algorithm in fact achieves a $\log |S_{max}|$-factor approximation algorithm for set cover on general instances where $|S_{max}|$ is the size of the largest set in the family. Thus, after the greedy algorithm picks $t$ sets each of which cover at least $\sqrt{n}$ new elements, in the remaining instance, we have $|S_{max}|\leq \sqrt{n}$. As there are $k$ sets that cover the remainining instance, the greedy algorithm picks at most $k \frac{\ln n}{2}$ sets after picking the $t$ sets. As each of the $t$ sets cover at least $\sqrt{n}$ new elements, $t\leq \sqrt{n}$. Overall, the total number of sets used by the greedy algorithm is equal to 
    \[
    t+k\frac{\ln n}{2} \leq \sqrt{n}+k\frac{\ln n}{2} \leq \left( 1+ \frac{\ln n}{2}\right) k
    \]
    \item Suppose that $k < \sqrt{n}$. In this case, using~(\ref{eq:set-size-bound}), we can infer that there are at most $k$ sets with size at least $k$ in the family. Thus, after the greedy algorithm picks $k$ sets, in the remaining instance, each set has size at most $k$, and thus, greedy algorithm picks at most $k \ln k$ sets. Overall, the total number of sets picked by the greedy algorithm is equal to \[
    k + k \ln k \leq k + k \frac{\ln n}{2} = \left( 1 + \frac{\ln n}{2}\right) k
    \]
\end{enumerate}
Thus, in both the cases, the greedy algorithm picks at most $k\left( 1+\frac{\ln n}{2}\right)$ sets. 

\medskip\noindent \textbf{A hard instance for the greedy algorithm.} We now prove that the above bound is tight for the greedy algorithm. Fix a large integer $k$, and let $n=k^2$, $X=[n]$. 
We first add $k$ sets to the family $\mathcal{S}$ $S_1, S_2, \ldots, S_k$ where $S_j = \{ (j-1)k+1, (j-1)k+2,\ldots, jk\}$.
Note that these $k$ sets together cover the whole universe $X$.
We view the universe $X$ as $k$ blocks, with the $j$'th block comprising of the set $S_j$. 

We now add $m=(k-1)\ln k$ additional pairwise disjoint sets $T_1, T_2, \ldots, T_m$ to $\mathcal{S}$ such that the greedy algorithm picks the set $T_i$ in the $i$'th iteration.  We choose the sets $T_i, i \in [m]$, as follows: 
\begin{enumerate}
    \item For $j \in [k]$, let the set $X_j$ be the uncovered elements of the block $S_j$. Let $a_j=|X_j|$ for all $j \in [k]$. We initially set $X_j = S_j$ for all $j \in [k]$. 
    \item At every iteration $i \in [m]$: 
    \begin{enumerate}
        \item Sort the elements $\{a_1, a_2, \ldots, a_k\}$ such that $a_{\alpha_1} \geq a_{\alpha_2}\geq \ldots \geq a_{\alpha_k}$ where $(\alpha_1, \alpha_2, \ldots, \alpha_k)$ is a permutation of $[k]$. Let $p=a_{\alpha_1} \leq k$ be the largest number of uncovered elements in a block. 
        \item Let $P=\{\alpha_1, \alpha_2, \ldots, \alpha_p\}$. For $l \in [p]$, let $u_l \in [n]$ be equal to the largest element in $X_{\alpha_l}$. 
        \[
            u_l = \max \{ b \mid b \in X_{\alpha_l}  \}
        \] 
        We set 
        \[
        T_i = \{ u_l \mid l \in [p]\}
        \]
        Furthermore, we set $X_{\alpha_l} = X_{\alpha_l}\setminus u_l$ for all $l \in [p]$. We also update $a_j, j \in [k]$ as $a_j = |X_j|$ for all $j \in [k]$. 
    \end{enumerate}
  \end{enumerate}
  
    In the above procedure to output the sets $T_i, i \in [m]$, the cardinality of $|T_i|\geq |T_{i+1}|$ for all $i$. Furthermore, in the $i$th iteration of the above procedure, the cardinality of $T_i$ is at least the number of elements in any block that are not covered yet. This ensures that the greedy algorithm in the $i$th iteration picks the set $T_i$.
    Furthermore, as the sets $T_i$s are all mutually disjoint, and intersect each block at most once, the resulting set system is indeed a simple set system. 
    
    Our goal is to prove that after all the $m$ sets are picked, there are still uncovered elements in $[n]$. 
    For an integer $i \in [m]$, we let $s_i \in [n]$ denote the number of elements not covered by the greedy algorithm before the set $T_i$ is picked.
    For $i \in [m]$, the size of the set $T_i$ picked by the greedy algorithm in the $i$th iteration is equal to the largest number of uncovered elements in a block i.e. the value of $a_{\alpha_1}$ in the $i$th iteration. Based on the updating procedure followed above, we can infer that this value is equal to $|T_i|=\left \lceil \frac{s_i}{k}\right\rceil $. This follows from the fact that at any iteration of the above procedure, the sorted values $a_{\alpha_1}$ and $a_{\alpha_k}$ satisfy $a_{\alpha_1}\leq a_{\alpha_k}+1$.
 
    We have $s_1 = n = k^2$, and 
    \[
    s_{i+1} = s_{i}-|T_i| = s_{i}-\left\lceil \frac{s_i}{k} \right\rceil \geq s_i \left( 1 - \frac{1}{k} \right) -1 
    \]
    By setting $t_i = s_i +k$ for $i \in [m]$, we get 
    \[
    t_{i+1} \geq t_i \left( 1 - \frac{1}{k} \right) 
    \]
    Thus, we get 
    \begin{align*}
    t_{m+1}&\geq t_1 \left( 1 - \frac{1}{k} \right)^m \\
    &= (k^2+k) \left( 1 - \frac{1}{k} \right)^{(k-1) \ln k}\\
    &\geq (k^2+k) \exp \left(-\frac{\frac{1}{k}}{1-\frac{1}{k}}(k-1) \ln k \right) \quad (\text{Using }1-x \geq e^{\frac{-x}{1-x}}\, \forall \,0 \leq x <1)\\
    &= k+1
    \end{align*}
    Thus, $s_{m+1} \geq 1$, which proves that there are elements that are not covered after the greedy algorithm uses $m=(k-1)\ln k$ sets. 
    This completes the proof that there are simple set systems on $n$ elements with $k=\sqrt{n}$ sets covering all the elements where as the greedy algorithm picks at least $(k-1)\ln k\geq k\left( \frac{\ln n}{2}-1 \right)$ sets. \qedhere

\end{proof}

Turning to the hardness of set cover on simple set systems, we note that the inapproximability result of~\Cref{thm:hypergraph-vc} holds for $t=(\ln n)^{c}$ for an absolute constant $c>0$, with a weaker completeness guarantee of a vertex cover of measure $\approx 2/t$. Combining it with our reduction in the proof of~\Cref{thm:simple}, and recalling that the dual of a simple hypergraph is also simple, we can deduce that set cover on simple set systems is NP-hard to approximate to a factor better than $(\ln n)^{\Omega(1)}$. It remains an interesting question to determine the exact approximability of set cover on simple set systems. 